\pdfoutput=1
\documentclass[a4paper,11pt]{article}
\usepackage[colorlinks=true,citecolor=OliveGreen,linkcolor=BrickRed,urlcolor=BrickRed]{hyperref}
\usepackage{amssymb,amsmath,amsthm,mathtools}
\usepackage[cm]{fullpage}
\usepackage{cite}
\usepackage{setspace}
\usepackage[usenames,dvipsnames,svgnames,table]{xcolor}
\definecolor{darkgreen}{rgb}{0.0,0,0.9}
\usepackage[numbers]{natbib}
\usepackage[]{framed}
\usepackage{tikz}
\usepackage{caption}
\usepackage{subcaption}
\usepackage{color,colortbl}
\usetikzlibrary{arrows,shapes}
\usetikzlibrary{positioning}
\usepackage{tkz-graph}
\usetikzlibrary{decorations.pathmorphing}
\newcommand{\LL}{\mathbf{L}}
\newcommand{\PP}{\mathbf{P}}
\newcommand{\HH}{\mathbf{H}}
\newcommand{\bHH}{\bar{\mathbf{H}}}
\newcommand{\pp}{\mathbf{p}}
\newcommand{\AAA}{\mathbf{A}}
\newcommand{\uu}{\mathbf{u}}
\newcommand{\vv}{\mathbf{v}}
\newcommand{\xx}{\mathbf{x}}
\newcommand{\VV}{\mathbf{V}}
\newcommand{\Vcal}{\mathcal{V}}
\newcommand{\Ecal}{\mathcal{E}}
\newcommand{\GG}{\mathcal{G}}
\newcommand{\UU}{\mathbf{U}}
\newcommand{\aaa}{\mathbf{a}}

\newcommand{\Exp}{\mathbb{E}}

\newcommand{\Rset}{\mathbb{R}}

\DeclareSymbolFont{sfoperators}{OT1}{cmss}{m}{n}
\DeclareSymbolFontAlphabet{\mathsf}{sfoperators}
\makeatletter
\def\operator@font{\mathgroup\symsfoperators}
\makeatother

\DeclareMathOperator{\Cov}{Cov}

\DeclareMathOperator{\rankk}{rank}
\DeclareMathOperator{\bern}{Bern}
\DeclareMathOperator{\diag}{diag}
\newtheorem{theorem}{Theorem}[]

\newtheorem{lemma}{Lemma}[]
\newtheorem{remark}{Remark}[]

\author{Kasra Khosoussi\\[0.16cm]\href{http://kasra.github.io}{kasra.github.io}\\\texttt{kasra.khosoussi@uts.edu.au}}
\date{September 20, 2018}
\title{On the Expected Value of the Determinant of Random Sum of Rank-One Matrices}
\begin{document}
\maketitle
\begin{abstract}
  We present a simple, yet useful result about the expected value of the determinant
  of random sum of rank-one matrices. Computing such expectations in general
  may involve a sum over exponentially many terms. Nevertheless, we show that an
  interesting and useful class of such expectations that arise in, e.g.,
  D-optimal estimation and random graphs can be computed efficiently via
  computing a single determinant.
\end{abstract}
\section{Problem Definition}
\begin{itemize}
  \item $[n] \triangleq \{1,2,\dots,n\}$, and for any finite set $\mathcal{W}$,
	$\binom{\mathcal{W}}{k}$ is the set of $k$-subsets of $\mathcal{W}$.
  \item Suppose we are given a pair of $m$ real $n$-vectors,
	$\{\uu_i\}_{i=1}^{m}$ and $\{\vv_i\}_{i=1}^m$. Define,
	\begin{equation}
	  \UU \triangleq \big[
		\uu_1 \,\, \uu_2 \,\, \cdots \,\, \uu_m
	  \big] \qquad
	  \VV \triangleq
		\big[
		\vv_1 \,\, \vv_2 \,\, \cdots \,\, \vv_m
	  \big]
	  \label{}
	\end{equation}
  \item Let $\{\pi_i\}_{i=1}^{m}$ be $m$ \emph{independent} Bernoulli random
	variables distributed as,
	\begin{align}
	  \pi_i &\sim \bern(p_i) \quad\quad \text{$i \in [m]$} \\
	  \pi_i & \perp \pi_j \quad\quad\qquad\,\,  i,j\in[m], i\neq j
	  \label{}
	\end{align}
	where $\{p_i\}_{i=1}^{m}$ are given. Define $\pp \triangleq [p_1 \,\, p_2
	\,\, \cdots \,\, p_m]^\top$ and $\boldsymbol\pi \triangleq 
	  [\pi_1 \, \pi_2 \, \cdots \, \pi_m]^\top.
	$
  \item We are interested in computing the expression below,
	\begin{align}
	  e(\UU,\VV,\pp) & \triangleq \Exp_{\boldsymbol\pi}\,\Big[\det\Big(\sum_{i=1}^{m} \pi_i \uu_i\vv_i^\top
	  \Big)\Big] \\ &= 
	  \Exp_{\boldsymbol\pi}\,\Big[\det\Big(\UU \Pi \VV^\top \Big)\Big]
	  \label{}
	\end{align}
	where $\Pi \triangleq \diag(\pi_1,\pi_2,\dots,\pi_m)$. 
	Note that the naive way of computing this expectation leads to a
	computationally	intractable sum over
	$\{0,1\}^{m}$.
\section{Main Result}
\begin{theorem}[Main Result \cite{kasraArxiv16}]
	\begin{align}
	  e(\UU,\VV,\pp) & = \det\Big(\sum_{i=1}^{m} p_i \uu_i\vv_i^\top
	  \Big) \\ &= 
	  \det\Big(\UU \PP \VV^\top \Big),
	  \label{}
	\end{align}
	where $\PP \triangleq \diag(p_1,p_2,\dots,p_m)$.
	  \label{th:main}
	\end{theorem}
\begin{proof}[Proof of Theorem 1] The proof outline is as follows:
  \begin{enumerate}
	\item[Step 1.] First, the Cauchy-Binet formula is used to expand the determinant
	  as a sum over $\binom{m}{n}$ terms.
	\item[Step 2.] The expected value of each of the $\binom{m}{n}$ terms can be easily computed. 
	\item[Step 3.] Finally, the Cauchy-Binet formula
	  is applied again to shrink the sum.
  \end{enumerate}
  We begin by applying the Cauchy-Binet formula:
\begin{align}
  \Exp_{\boldsymbol\pi}\,\Big[\det\Big(\sum_{i=1}^{m} \pi_i \uu_i\vv_i^\top \Big)\Big]
  &=
  \Exp_{\boldsymbol\pi}\,\Big[\sum_{\mathcal{Q} \in
	\binom{[m]}{n}}\det\Big(\sum_{k \in
	\mathcal{Q}} \pi_k \uu_k\vv_k^\top \Big)\Big] \label{eq:cb} \\
  &=
  \sum_{\mathcal{Q} \in
  \binom{[m]}{n}}\Exp_{\boldsymbol\pi}\,\Big[\det\Big(\sum_{k \in
  \mathcal{Q}} \pi_k \uu_k\vv_k^\top \Big)\Big].
  \label{eq:final}
\end{align}
Since $|\mathcal{Q}| = n$,  we will have $\rankk\Big(\sum_{k \in
  \mathcal{Q}} \pi_k \uu_k\vv_k^\top\Big) < n$ if there exists $k \in
  \mathcal{Q}$ for which $\pi_k = 0$.
  Hence, the determinant can be non-zero only when 
  $\pi_k = 1$ for all $k \in \mathcal{Q}$. Therefore,
\begin{align}
  \det\Big(\sum_{k \in
	\mathcal{Q}} \pi_k \uu_k\vv_k^\top\Big) = 
	\begin{cases}
	\det\Big(\sum_{k \in
	\mathcal{Q}} \uu_k\vv_k^\top\Big)   & \text{iff $\pi_k = 1$ for all $k \in \mathcal{Q}$,} \\
	0 & \text{otherwise.}
	\end{cases}
  \end{align}
  But from the independence assumption we know that,
\begin{equation}
  \mathbb{P}\,\Big[\bigwedge_{k \in \mathcal{Q}} \pi_k = 1\Big] =
  \prod_{k \in \mathcal{Q}} p_k.
  \label{}
\end{equation}
Each individual expectation in \eqref{eq:final} can be computed as follows.
\begin{align}
  \Exp_{\boldsymbol\pi}\,\Big[\det\Big(\sum_{k \in
  \mathcal{Q}} \pi_k \uu_k\vv_k^\top \Big)\Big] 
  & =  \det\Big(\sum_{k \in
  \mathcal{Q}} \uu_k\vv_k^\top \Big) 
  \,\prod_{k \in \mathcal{Q}} p_k \\
  & = 
  \det\Big(\sum_{k \in
  \mathcal{Q}} p_k \uu_k\vv_k^\top \Big).
  \label{eq:simplified}
\end{align}
Plugging \eqref{eq:simplified} back into \eqref{eq:final} yields,
\begin{align}
  \Exp_{\boldsymbol\pi}\,\Big[\det\Big(\sum_{i=1}^{m} \pi_i \uu_i\vv_i^\top \Big)\Big]
  &= 
  \sum_{\mathcal{Q}\in\binom{[m]}{n}}\det\Big(\sum_{k \in
  \mathcal{Q}} p_k \uu_k\vv_k^\top \Big).
  \label{eq:almost}
\end{align}
Note that \eqref{eq:almost} is
nothing but the Cauchy-Binet expansion of $\det\Big(\sum_{i=1}^m p_i \uu_i\vv_i^\top
\Big)$. This concludes the proof.
\end{proof}
%
\end{itemize}
\section{Motivation \& Applications}
$e(\UU,\VV,\pp)$ arises in the following problems:
\begin{enumerate}
  \item \textbf{Estimation}
	\\
	Suppose $\xx \in \Rset^n$ is an unknown quantity to be estimated
	using $m$ observations $\{z_i\}_{i=1}^m$ ($m \geq n$) generated according to,
	\begin{align}
	  \mathbf{z} = \mathbf{H}\mathbf{x} + \epsilon \quad \text{where} \quad
	  \epsilon \sim \mathcal{N}(\mathbf{0},\Sigma)
	  \label{}
	\end{align}
	where $\mathbf{z}\triangleq [z_1 \,\, z_2 \,\, \cdots \,\, z_m]^\top$.
	To simplify our notation, let us define  $\bar{\HH} \triangleq
	\Sigma^{-1/2}\HH$. The maximum likelihood estimator $\hat{\xx}$ has the
	following form:
	\begin{equation}
	  \hat{\xx} = {(\bHH^\top\bHH)}^{-1}\bHH^\top\mathbf{z}.
	\end{equation}
	It is well known that $\hat{\xx}$ is unbiased and \emph{efficient}; i.e., it
	achieves the Cram\'{e}-Rao lower bound,
	\begin{equation}
	  \Cov\,[\hat{\xx}] = {\big(\bHH^\top\bHH\big)}^{-1}.
	  \label{}
	\end{equation}
	
	Geometrically speaking, the hypervolume of uncertainty hyperellipsoids are
	proportional to $\sqrt{\det\Cov\,[\hat{\xx}]}$ (see, e.g.,
	\cite{joshi2009sensor}).
	The D-optimality (determinant-optimality) criterion is
	defined as $\det\Cov\,[\hat{\xx}]^{-1}$. 
	Note that $\det\Cov\,[\hat{\xx}] = {(\det\,\mathcal{F})}^{-1}$
	where $\mathcal{F} \triangleq \bHH^\top\bHH$ is the so-called Fisher
	information matrix. Hence, minimizing the determinant of the estimation
	error covariance matrix is
	equivalent to maximizing the D-optimality criterion, $\det\,(\bHH^\top\bHH)$.
	Now consider the following scenarios.
	\begin{enumerate}
	  \item \textbf{Sensor Failure}: The $i$th ``sensor'' may ``fail''
		independently with probability $1-p_i$, for all $i \in [m]$. In this case, the
		row corresponding to each failed sensor has to be removed from
		$\bHH$.
		Hence, $e(\bHH^\top\hspace{-0.09cm},\bHH,\pp)$ gives the \emph{expected
		value} of the D-optimality criterion.
  \item \textbf{Sensor Selection}: 
	The goal in D-optimal sensor selection is to select a subset (e.g.,
	$k$-subset) of $m$ available sensors (observations) such that the
	D-optimality criterion is maximized. 
	Joshi and Boyd \cite{joshi2009sensor} proposed an approximate solution to
	this problem through convex relaxation. In \cite{kasraArxiv16}, we showed
	that their convex program can be interpreted as the problem of finding the optimal
	probabilities $\{p_i\}_{i=1}^{m}$ for randomly selecting (e.g., $k$) sensors
	via independent coin tosses such that the expected value of the D-optimality
	criterion,
	i.e., $e(\bHH^\top\hspace{-0.09cm},\bHH,\pp)$, is maximized. See
	\cite{kasraArxiv16,kasra16wafr} for the details.
	\end{enumerate}
	\begin{remark} For sufficiently smooth nonlinear measurement models, $\bHH$ should be
	  replaced by the normalized Jacobian of the measurement function.
	\end{remark}
	\begin{figure}[t]
	  \centering
		\begin{tikzpicture}
		  \GraphInit[vstyle=Classic]
		  \SetUpVertex[FillColor = blue, MinSize=0.1cm]
		  \tikzset{VertexStyle/.append style = {
			font=\Large,
		  draw}}
		  \SetGraphUnit{2}
		  \SetVertexMath
		  \begin{scope}[]
			\Vertices{circle}{v_1,v_2,v_3,v_4,v_5}
		  \end{scope}
		  \SetUpEdge[lw = 1pt,
			color = black,
			labelcolor = white,
		  labelstyle = {draw,text=black}]
		  \Edge[label = $p_{1}$](v_1)(v_2)
		  \Edge[label = $p_{2}$](v_2)(v_3)
		  \Edge[label = $p_{3}$](v_3)(v_4)
		  \Edge[label = $p_{4}$](v_5)(v_4)
		  \Edge[label = $p_{6}$](v_5)(v_2)
		  \Edge[label = $p_{5}$](v_1)(v_5)
		\end{tikzpicture}
		\caption{A random edge-weighted graph $\GG_\pi$ with probabilities
		  $\{p_i\}_{i=1}^{6}$. The edge weights (not shown here) are
		  assigned by $w : \Ecal \to \Rset_{>0}$.}
		\label{fig:graph}
	\end{figure}
	\pagebreak
  \item \textbf{Spanning Trees in Random Graphs\footnote{We first presented Theorem~\ref{th:main}, and its special case used
for computing the weighted number of spanning trees, in
\cite{kasraArxiv16}. Recently
we discovered an earlier result for computing the expected number of spanning
trees in unweighted anisotropic random graphs by Joel E. Cohen in 1986
\cite{cohen1986connectivity}. Cohen in \cite{cohen1986connectivity} provides a
different proof and extends his result to the case of random directed graphs.
Our result, however, considers the weighted graphs, while our Theorem~\ref{th:main}
extends it to the general case of random sum of arbitrary rank-one matrices.}
	}
	\\
	Networks with ``reliable'' (against, e.g., noise in estimation, or failure
	in communication) topologies are crucial in many applications across science
	and engineering. In general, the notion of reliability in networks is
	closely related to graph connectivity.  Among the existing combinatorial and
	spectral graph connectivity criteria, the number of spanning trees
	(sometimes referred to
	as \emph{graph complexity} or \emph{tree-connectivity}) stands out: despite
	its combinatorial origin, it can also be characterized solely by the
	spectrum of the graph Laplacian (Kirchhoff) matrix. This result is due to
	Kirchhoff's matrix-tree theorem (and its extensions):
	\begin{theorem}[Kirchhoff's Matrix-Tree Theorem for Weighted Graphs]
	  Consider graph $\GG = (\Vcal,\Ecal,w)$ where $\Vcal =
	  \{v_i\}_{i=0}^{n}$,
	  $\Ecal \subseteq \binom{\Vcal}{2}$, and $w : \Ecal \to \Rset_{>0}$.
	  The \emph{reduced Laplacian matrix of $\GG$}, denoted by $\LL_\GG$, is
	  obtained by removing an arbitrary row and the corresponding column from
	  the (weighted) Laplacian matrix of $\GG$; e.g., $v_0$.
	  The weighted number of spanning is given by,
	  \begin{align}
		t_w(\GG) &\triangleq \sum_{\mathcal{T} \in \mathbb{T}(\GG)}
		\prod_{e \in \Ecal(\mathcal{T})} w(e) \\
		& = \det\,(\LL_\GG)
		\label{}
	  \end{align}
	  where $\mathbb{T}(\GG)$ is the set of all spanning trees of $\GG$, and
	  $\Ecal(\mathcal{T})$ denotes the edge set of graph $\mathcal{T}$. 
	  Note that in case of unit weights, $t_w(\GG)$ is simply the number of
	  spanning trees in $\GG$.
	  \label{th:mt}
	\end{theorem}
	
	Now consider a random graph whose $i$th edge is ``operational'' with
	probability $p_i$, independent of other edges
	(Figure~\ref{fig:graph}).\footnote{Here, ``operational'' means that the
	corresponding vertices are connected via an edge.}
	Define indicator variables $\{\pi_i\}_{i=1}^{m}$ such that $\pi_i = 1$ iff
	the $i$th edge is operational, otherwise $\pi_i = 0$.
	The \emph{reduced (unweighted) incidence} matrix of $\GG$, $\AAA =
	[\aaa_1 \, \aaa_2 \, \cdots \, \aaa_m]$, is
	obtained by removing an arbitrary row from the (unweighted) incidence matrix
	of $\GG$. From Theorem~\ref{th:mt} we know that,
	\begin{align}
	  \Exp_{\pi}\,\Big[ t_w(\GG_\pi) \Big] = 
	  \Exp_{\pi}\,\Big[ \det\Big( \sum_{i=1}^{m} \pi_i w(e_i)\,
	  \aaa_i\aaa_i^\top\Big) \Big].
	  \label{}
	\end{align}
	Define $\AAA_w \triangleq \AAA\sqrt{\mathbf{W}}$ in which $\mathbf{W}
	\triangleq \diag\big(w(e_1) \, w(e_2) \, \cdots \, w(e_m)\big)$. Note that
	this expression is equal to $e(\AAA_w,\AAA_w^\top,\pp)$. From
	Theorem~\ref{th:main} we have,
	\begin{align}
	  \Exp_{\pi}\,\Big[ t_w(\GG_\pi) \Big] &= 
	  \Exp_{\pi}\,\Big[ \det\Big( \sum_{i=1}^{m} \pi_i w(e_i)\,
	  \aaa_i\aaa_i^\top\Big) \Big] \\ 
	  &= e(\AAA_w,\AAA_w^\top,\pp) \\
	  & = 
	  \det\Big( \sum_{i=1}^{m} p_i w(e_i)\,\aaa_i\aaa_i^\top\Big) \\
	  & = \sum_{\mathcal{T} \in \mathbb{T}(\GG)}
		\prod_{e_i \in \Ecal(\mathcal{T})} p_i w(e_i).
	  \label{eq:pq}
	\end{align}
	\begin{remark}
	  It is worth mentioning that, according to above equations, the expected
	  weighted number of spanning trees is given by computing the weighted number
	  of spanning trees after multiplying the edge weights by their probabilities;
	  i.e.,
	  \begin{align}
		\Exp_{\pi}\,\Big[ t_w(\GG_\pi) \Big] &= t_{w_p}(\GG),
	  \end{align}
	  where $w_p : e_i \mapsto p_i w(e_i)$.
	\end{remark}
\end{enumerate}

\section{Random Sum of Rank-$r$ Matrices}
It is not immediately clear whether there is an efficient way for computing
\begin{align}
  \Exp_{\boldsymbol\pi}\,\Big[\det\Big(\sum_{i=1}^{m} \pi_i \UU_i\VV_i^\top
	  \Big)\Big]
  \label{}
\end{align}
in which $\UU_i$ and $\VV_i$ belong to $\Rset^{n \times r_i}$ for $i \in [m]$.
Nevertheless, the following results provide some preliminary insights into this
more general case.
The proofs of the following lemmas follow that of Theorem~\ref{th:main}---i.e.,
Cauchy-Binet formula.
\begin{lemma}
  \begin{align}
	\Exp_{\boldsymbol\pi}\,\Big[\det\Big(\sum_{i=1}^{m} \pi_i \UU_i\VV_i^\top
	\Big)\Big] & \geq \det\Big(\sum_{i=1}^m p_i \UU_i\VV_i^\top \Big).
	\label{eq:lm1}
  \end{align}
\end{lemma}
\begin{lemma}
 Consider a random graph $\GG_\pi$ (over graph $\GG$) whose edge set $\Ecal$ is
 partitioned into $k$ blocks $\{\Ecal_i\}_{i=1}^{k}$. The edges in the $i$th
 block are operational, independent of other blocks, with probability $p_i$.
 Let $\AAA_i$ be the collection of the columns of the reduced weighted 
 incidence matrix that belong to the $i$th block of edges $\Ecal_i$. We have,
	\begin{align}
	  \Exp_{\pi}\,\Big[ t_w(\GG_\pi) \Big] &= 
	  \Exp_{\pi}\,\Big[ \det\Big( \sum_{i=1}^{m} \pi_i \,
	  \AAA_i\AAA_i^\top\Big) \Big] \\ 
	  & = \sum_{\mathcal{T} \in \mathbb{T}(\GG)}
	  \prod_{e_i \in \Ecal(\mathcal{T})} p_{b_i}^{1/n_{b_i}(\mathcal{T})} w(e_i)
	\end{align}
	where $b_i$ is the block index that contains $e_i$  
	and $n_{i}(\mathcal{T}) \triangleq |\Ecal(\mathcal{T})
	\cap \Ecal_{i}|$.
  \label{}
\end{lemma} 
\pagebreak
\bibliographystyle{../wafr16/paper/splncs_srt}
\bibliography{../rss16-trees/paper/graph,../bib/slam}

\begin{thebibliography}{1}

\bibitem{cohen1986connectivity}
Cohen, J.E.:
\newblock Connectivity of finite anisotropic random graphs and directed graphs.
\newblock In: Mathematical Proceedings of the Cambridge Philosophical Society.
  Volume~99., Cambridge Univ Press (1986)  315--330

\bibitem{joshi2009sensor}
Joshi, S., Boyd, S.:
\newblock Sensor selection via convex optimization.
\newblock Signal Processing, IEEE Transactions on \textbf{57}(2) (2009)
  451--462

\bibitem{kasra16wafr}
Khosoussi, K., Sukhatme, G.S., Huang, S., Dissanayake, G.:
\newblock Designing sparse reliable pose-graph {SLAM}: A graph-theoretic
  approach.
\newblock International Workshop on the Algorithmic Foundations of Robotics
  (2016)

\bibitem{kasraArxiv16}
Khosoussi, K., Sukhatme, G.S., Huang, S., Dissanayake, G.:
\newblock Maximizing the weighted number of spanning trees: Near-$t$-optimal
  graphs.
\newblock arXiv:1604.01116 (2016)

\end{thebibliography}
\end{document}